\documentclass{amsart}
\usepackage[utf8]{inputenc}

\usepackage{xcolor}
\usepackage{amsmath}
\usepackage{amssymb}
\usepackage{comment}
\usepackage{hyperref}
\usepackage{graphicx}
\usepackage{subfigure}
\usepackage{bm}
\usepackage{geometry}
\usepackage{tikz}
\usetikzlibrary{decorations.pathreplacing}
\usetikzlibrary{fadings}

\allowdisplaybreaks

\title{Pricing options on flow forwards by neural networks in Hilbert space}
\author{Fred Espen Benth, Nils Detering, Luca Galimberti}
\address{Fred Espen Benth \\
University of Oslo\\
Department of Mathematics \\
P.O. Box 1053, Blindern\\
N--0316 Oslo, Norway}
\email[]{fredb\@@math.uio.no}
\address{Nils Detering \\ 
University of California at Santa Barbara\\
Department of Statistics and Applied Probability\\
CA 93106 Santa Barbara, USA}
\email[]{detering\@@pstat.ucsb.edu}
\address{Luca Galimberti \\ 
Norwegian University of Science and Technology\\
Department of Mathematical Sciences\\
Sentralbygg 2, Gl\o shaugen, Trondheim, Norway}
\email[]{luca.galimberti\@@ntnu.no}

\newtheorem{theorem}{Theorem}[section]
\newtheorem{definition}[theorem]{Definition}
\newtheorem{lemma}[theorem]{Lemma}
\newtheorem{proposition}[theorem]{Proposition}

\newtheorem{remark}[theorem]{Remark}
\newtheorem{example}[theorem]{Example}

\newcommand{\R}{\mathbb R}

\newcommand{\N}{\mathbb N}

\renewcommand{\P}{\mathbb P} 

\newcommand{\norm}[1]{\left\lVert#1\right\rVert}
\newcommand{\abs}[1]{\left |#1\right|}
\newcommand{\X}{\mathfrak X}

\DeclareMathOperator{\Span}{span}
\DeclareMathOperator{\Supp}{supp}

\newcommand{\luca}[1]{\textcolor{orange}{#1}}

\date{\today}
\thanks{Fred  Espen  Benth  acknowledges  support  from  SPATUS,  a  Thematic  Research  Group funded by UiO:Energy. \\
Luca Galimberti has been supported in part by the grant Waves and Nonlinear Phenomena (WaNP) from the Research Council of Norway.
}

\begin{document}

\maketitle
\begin{abstract}
 We propose a new methodology for pricing options on flow forwards by applying infinite-dimensional neural networks. We recast the pricing problem as an optimization problem in a Hilbert space of real-valued function on the positive real line, which is the state space for the term structure dynamics. This optimization problem is solved by facilitating a novel feedforward neural network architecture designed for approximating continuous functions on the state space. The proposed neural net is built upon the basis of the Hilbert space. We provide an extensive case study that shows excellent numerical efficiency, with superior performance over that of a classical neural net trained on sampling the term structure curves. 
\end{abstract}
\section{Introduction}

In commodity markets, options are typically written on forward and futures contracts (see e.g. Geman \cite{Geman}). In some markets, like for example electricity and gas, as well as freight and weather markets on temperature and wind, the forwards deliver the underlying commodity or service over a contracted delivery period, and not at a specified delivery time in the future (see e.g. Benth, \v{S}altyt\.e Benth and Koekebakker \cite{BSBK-book}). Such forwards are sometimes referred to as {\it flow forwards}.  Also, futures contracts on the Secured Overnight Financing Rate (SOFR) fall into this class of flow forwards, as they are written on the average of the underlying SOFR over contracted periods of time (see Skov and Skovmand \cite{SkovSkovmand}).

There is a large literature on neural networks and financial derivatives, mostly focusing on approximating option prices. We refer to the recent survey by Ruf and Wang \cite{RufWang} for an extensive historical account of the various papers and their results and methodology. As is well-known in mathematical finance, option prices can be re-cast as solutions of partial differential equations based on Feynman-Kac formulae for diffusion processes (see e.g., Karatzas and Shreve \cite{KS} and Bj\"ork \cite{bjork}). This connection has been utilised in, say, Beck {\it et al.} \cite{Beck2021, beck2021overview}, Han, Jentzen and E \cite{Han8505} and Hutzenthaler {\it et al.} \cite{HJKN}
in studying deep neural network approximations. Gonon and Schwab \cite{Gonon2021} have in a recent paper extended this theory to a study of expression rates for deep neural network approximations of option prices on L\'evy-based asset prices. In all these papers, the main argument for introducing deep neural networks is to overcome the curse of dimensionality. Thus, deep neural networks can be applied to price options on a high-dimensional underlying, like basket options, say. 

In this paper we bring this to the "ultimate" high-dimensional case, considering deep neural networks approximating option prices on {\it infinite-dimensional} underlyings. Indeed, option prices on flow forwards are in general functions of functions, as the underlying will be a curve (i.e., the term structure) rather than a vector of points (i.e., prices of the underlying assets). We propose to approximate this non-linear option price functional by a neural network in Hilbert space.       


Facilitating neural network in the context of flow forwards has first been proposed in Benth, Detering and Lavagnini \cite{BDL2021} where a classical neural network has been used for calibrating parameters.
Motivated by the approach of Beck {\it et al.} \cite{Beck2021}, we generalise the option pricing problem to deep neural networks in Hilbert space, appealing to the general neural nets in Fr\'echet space and their universal approximation of continuous mappings developed in Benth, Detering and Galimberti \cite{benth2021neural}.  
We also mention in passing the recent interest in applying neural networks to compute the implied volatility in connection with rough volatility models, see Bayer and Stemper \cite{BayerStemper} and  Bayer {\it et al.} \cite{Horvath}. Rough volatility models are non-Markovian, however, can be viewed as infinite-dimensional objects. We refer to Benth, Eyjolfsson and Veraart \cite{BEV-SIFIN} for more on this latter perspective for general Volterra processes and pricing in energy markets.   

To be more specific about the problem we are dealing with and the motivation for our approach, we discuss briefly options on power. In electricity markets, such as the European Energy Exchange (EEX) and Nord Pool one can trade in call and put options written on forward contracts delivering power over a contracted period of time. I.e., for a contracted delivery period $[T_1,T_2]$, where $0\leq T_1<T_2$, we denote the price at time $t\leq T$ by $\hat{F}(t,T_1,T_2)$. The delivery period for forward contracts underlying an option are typically a month, say July, so that for example $T_1=July1$ and $T_2=July 31$. For a strike price $\mathcal K$ and exercise time $\tau\leq T_1$, the price of a call option at time $t\leq\tau$ is defined as
\begin{equation}\label{option:price}
    V(t,\tau)=\mathbb E[\max(\hat{F}(\tau,T_1,T_2)-\mathcal K,0)\,\vert\,\mathcal F_t]
\end{equation}
Here, we have chosen $r=0$, and assume that we are working directly under the martingale measure $\mathbb Q$ so that $V(t,\tau)$ is the arbitrage-free price. In the market place, a usual exercise time is $\tau=T_1$, and the strikes of the most liquid options are close to current forward prices. In this paper we shall not make any distinction between forwards and futures, and focus solely on the forward contracts. 

One would like to have a model for the forward prices across different delivery periods which avoids calendar arbitrage, that is, arbitrage from investing in forwards with different delivery periods. As discussed in Benth and Koekebakker \cite{BK}, the most convenient way to do this is via fixed-delivery forwards. To this end, assume that $F(t,T)$ is the price of a forward at time $t$ delivering the underlying commodity at time $T\geq t$. Of course, in power markets, such a forward only makes sense as a modeling device. The following relationship holds then true (see Benth, \v{S}altyt\.e Benth and Koekebakker \cite[Prop. 4.1.]{BSBK-book}):
\begin{equation}
\label{forward:price}
    \hat{F}(t,T_1,T_2)=\frac1{T_2-T_1}\int_{T_1}^{T_2}F(t,T)dT
\end{equation}
The forward price is denoted in power markets {\it per MWh}, which is the reason for dividing by the length of delivery, i.e., defining $\hat{F}$ as the average of $F$ over the delivery period.

If one now has a dynamic model for the forward curve $F(t,\cdot)$, then prices for related forwards and options can be obtained via the relationship (\ref{option:price}) and (\ref{forward:price}). Dynamic models for the forward curve have been proposed and analysed in, say, Clewlow and Strickland \cite{CS}, Benth and Koekebakker \cite{BK} and Benth and Kr\"uhner \cite{BK-CIMS}.
Such a dynamic is usually defined in terms of some stochastic partial differential equation (SPDE). However, while models for the forward curve lead to coherent arbitrage free prices across options with different delivery periods, they pose computational challenges because of their infinite dimensional nature. 

This paper tries to overcome some of these challenges by proposing a numerical method for pricing options on flow forwards based on neural networks. We first derive properties of the pricing function which ensure that one can actually recast the pricing problem (\ref{option:price}) into an optimization problem over a space of continuous functions defined on a Hilbert space of functions. We then show by a density argument that it is actually sufficient to optimize only over a restricted set of continuous functions, namely Hilbert space neural networks. These neural networks have been proposed in \cite{benth2021neural} for approximating functionals defined on a Fr\'echet space $\X$. We show how one can efficiently implement these neural networks by standard machine learning packages as TensorFlow/Keras and use the stochastic gradient descent algorithm for the optimization task. Our method delivers option prices automatically for a wide range of initial market conditions. This has the tremendous advantage that expensive simulations are performed only once for the training step of the neural network, and they do not have to be repeated if market conditions change. 

We test our methodology in some numerical case studies, and find that it works very well with high dimensional noise which is in line with the general perception that numerical methods based on neural networks can often overcome the course of dimensionality. In the case studies, we also compare our approach with the direct classical neural network approximation where the term structure curves are simply sampled and turned into high-dimensional input objects. Numerical evidence talks strongly in favour of our infinite-dimensional network in Hilbert space as being superior. Using the infinite-dimensional methodology we propose in this paper is also advantageous when approximating the Greeks of the option. As the network approximation is in terms of interpretable factors, like level, slope and curvature of the term structure, (numerical or analytical) differentiation of the network easily captures sensitivities with respect to these. Using a classical network trained by sampling, it is not easy to obtain approximations of the Greeks from the neural net. 

The outline of this paper is as follows: In Section~\ref{nn:frechet} we present and adapt results from \cite{benth2021neural} for the Hilbert spaces one typically deals with when pricing flow derivatives. In Section~\ref{sec:flow-forwards} we state two possible dynamics for the flow forward curve. In Section~\ref{sec:option} we derive required continuity properties of the pricing functional in both models. In Section~\ref{sec:approx:price:fct} we study the related optimization problem. Finally, Section~\ref{sec:case:study} contains an extensive numerical case study.

\section{An introduction to feedforward neural networks on Hilbert spaces}\label{nn:frechet}

As anticipated in the Introduction, our strategy is based on neural networks in infinite dimensional vector spaces $\X$ and on an abstract approximation result for continuous functions $f:\X\to \R$; these kind of architectures have been introduced in the paper Benth, Detering and Galimberti \cite{benth2021neural} with $\mathfrak{X}$ being a Fr\'echet space. Here we are going to briefly outline their definition and their most salient properties.

We recall that the classical universal approximation theorem shows that any continuous function from $\mathbb{R}^n$ to $\mathbb{R}$ can be approximated arbitrary well with a one layer neural network. More precisely, for a fixed continuous function $\sigma : \mathbb{R}\rightarrow \mathbb{R}$ and $a \in \mathbb{R}^n, \ell , b \in \mathbb{R}$, a {\em neuron} is a function $\mathcal{N}_{\ell , a,b}\in C(\mathbb{R}^n ; \mathbb{R})$ defined by $x\mapsto \ell \sigma (a^{\top}x +b) $. The universal approximation theorem then states conditions on the {\em activation function} $\sigma$ such that the linear space of functions generated by the neurons 
$$\mathfrak N(\sigma):=\Span \{ \mathcal{N}_{\ell , a,b}; \ell ,b \in \mathbb{R} , a\in \mathbb{R}^n \} $$ is dense with respect to the topology of uniform convergence on compacts. This means that for every $f\in C(\mathbb{R}^n; \mathbb{R})$ and compact subset $K\subset \mathbb{R}^n$ and a given $\varepsilon >0$, there exists $N\in \mathbb{N}$ and $\ell_i ,b_i \in \mathbb{R} , a_i\in \mathbb{R}^n$ for $i=1,\dots , N$ such that
$$ \sup_{x\in K}\abs{f(x) - \sum_{i=1}^{N} \mathcal{N}_{\ell_i , a_i,b_i}  (x)}< \varepsilon .$$

For the sake of concreteness (but we remark that the following results are proved in \cite{benth2021neural} in much wider generality) let $(H,\langle\cdot,\cdot\rangle)$ now be an arbitrary real Hilbert space (not necessarily separable at this stage), where $\langle\cdot,\cdot\rangle$ denotes its scalar product. Its topological dual $H^*$ can be identified via Riesz's isomorphism with $H$ itself, and therefore in the sequel with a slight abuse of notation the symbol $\langle\cdot,\cdot\rangle$ will denote both the scalar product of $H$ and the canonical pairing between $H^*$ and $H$.

In order to define the infinite-dimensional analogue of a neuron, $a^{\top}x +b$ is replaced by an affine function on $H$, the activation function $\sigma : \mathbb{R} \rightarrow \mathbb{R}$ by a function in $C(H;H)$, and the scalar $\ell$ by a continuous linear form.
For $\ell\in H^*, A\in \mathcal{L} (H), b\in H$ a neuron, $\mathcal{N}_{\ell,A,b}$ is then defined by 
$$\mathcal{N}_{\ell,A,b}(x)= \langle \ell ,\sigma (Ax +b)\rangle$$
and one asks for conditions on $\sigma:H \rightarrow H$ that ensure that $\mathfrak N(\sigma):=\Span \{ \mathcal{N}_{\ell , A,b}; \ell\in H^* ,A\in \mathcal{L} (H),b \in H  \} $ is dense in $C(H;\mathbb{R})$ under some suitable topology.

In \cite{benth2021neural}, the following separating property for the activation function $\sigma$, which can be seen as the infinite-dimensional counterpart to the well known sigmoidal property for functions from $\R$ to $\R$ (see \cite{Cybenko1989}), was introduced:
\begin{definition}{Separating property from Benth, Detering and Galimberti \cite{benth2021neural}:}\label{sigmoid}
There exist $\psi\in H^*\setminus\{0\}$ and $u_+,u_-,u_0\in H$ such that either $u_+ \notin \Span \{u_0,u_-\}$ or $u_- \notin \Span \{u_0,u_+ \}$ and such that 
\begin{equation}\label{eq: abstract condition on sigma}
\begin{cases}
\lim_{\lambda\to\infty} \sigma(\lambda x) = u_+, \text{ if } x\in \Psi_+\\
\lim_{\lambda\to\infty} \sigma(\lambda x) = u_-, \text{ if } x\in \Psi_-\\
\lim_{\lambda\to\infty} \sigma(\lambda x) = u_0, \text{ if } x\in \Psi_0\\
\end{cases}
\end{equation}
where we have set 
\begin{equation*}
    \Psi_+ =\{ x\in H; \langle\psi,x\rangle >0 \}, \quad \Psi_- =\{ x\in H; \langle\psi,x\rangle <0 \}
\end{equation*}
and $\Psi_0=\ker(\psi)$.
\end{definition}
We point out that as a particular case of the separating property one may choose $u_0=u_-=0$ and $u_+\neq 0$ for instance.

The following result shows the density of $\mathfrak N(\sigma)$ if the activation function $\sigma$ satisfies the separating property. 

\begin{theorem}{(Adapted from Benth, Detering and Galimberti \cite[Thm. 2.3 and 2.8]{benth2021neural})}\label{prop: density}
Let $(H,\langle\cdot,\cdot\rangle)$ be a real Hilbert space, and let $\sigma:H\to H$ be continuous, satisfying \eqref{eq: abstract condition on sigma} and with bounded range $\sigma(H)\subset H$. 
Then $ \mathfrak N(\sigma)$ is dense in $C(H;\mathbb R)$ when equipped with the topology of uniform convergence on compacts. In other words, given $f\in C(H;\mathbb R)$, then, for any compact subset $K$ of $H$, and any $\varepsilon>0$, there exists $ \sum_{m=1}^M \mathcal{N}_{\ell_m,A_m,b_m}\in  \mathfrak N(\sigma)$ with suitable $ \ell_m\in H^*,A_m\in\mathcal{L}(H)$ and $b_m\in H$ such that
\begin{equation*}\label{approx:prop}
    \sup_{x\in K}\abs{f(x) - \sum_{m=1}^M  \mathcal{N}_{\ell_m,A_m,b_m}(x)}  < \varepsilon.
\end{equation*}
\end{theorem}

Moreover, the following result, ensuring that one can approximate a given abstract neural net arbitrary well via a neural network that is constructed from finite dimensional maps and that can thus be trained, is valid, as soon as one imposes that the Hilbert space $ H$ is separable: for an in-depth discussion about this condition, we refer to Benth, Detering and Galimberti \cite{benth2021neural}.  

\begin{proposition}{(Adapted from Benth, Detering and Galimberti \cite[Prop. 4.1]{benth2021neural})}\label{prop: finite dimensional approx, Banach}\label{finite:dim:projection:NN}
Let $(H,\langle\cdot,\cdot\rangle)$ be a real separable Hilbert space and let $(e_k)_{k\in\N}$ be a an orthonormal basis for $ H$. For each $N\in\N$ let
$$
 \Pi_N:  H \to \Span\{e_1,\dots ,e_N\}
$$
be the orthogonal projection on the first $N$ elements of the basis. Let $\sigma: H\to H$ be Lipschitz. Let $f\in C( H;\R)$, $K\subset H$ compact and $\varepsilon>0$. Assume
\begin{equation*}
    \mathcal N^{\epsilon} (x) = \sum_{j=1}^M\langle \ell_j,\sigma(A_jx+b_j)\rangle,\quad x\in H
\end{equation*}
with $\ell_j\in H^*,A_j\in\mathcal{L}( H)$ and $b_j\in H$ such that 
\begin{equation*}
    \sup_{x\in K}\abs{f(x)-\mathcal N^{\epsilon}(x)}<\varepsilon.
\end{equation*}
Fix $\delta>0$. Then there exists $N_\ast=N_\ast(\mathcal N^\epsilon,\delta)\in\N$ such that for $N\geq N_\ast$
\begin{equation}\label{approx:finite}
    \sup_{x\in K}\abs{f(x)-\sum_{j=1}^M\langle \ell_j\circ\Pi_N,\sigma(
    \Pi_{N}A_j\Pi_{N}x+\Pi_{N}b_j)\rangle}<\varepsilon+\delta.
\end{equation}
\end{proposition}

We mention that the function $\mathcal{N}^{\varepsilon}:  H\rightarrow \R$, which is required in the proposition above, exists for instance in view of Theorem \ref{prop: density}, as soon as one assumes additionally that $\sigma$ satisfies \eqref{eq: abstract condition on sigma} and has bounded range $\sigma( H)\subset  H$. 
Observe that if $\sigma$ is Lipschitz continuous, then every $g\in\mathfrak{N}(\sigma)$ is also Lipschitz continuous.
\begin{remark}
The terms appearing in the sum in (\ref{approx:finite}) can now easily be programmed in a computer. We see that for large $N$, it is sufficient to consider the finite dimensional input values $\Pi_N (x)$ instead of $x$, and then successively the restriction of the operators $\Pi_N A_j, \sigma$ and $\ell_j$ to $\Span\{e_1,\dots ,e_N\}$ instead of the maps $A_j, \sigma$ and $\ell_j$ for $j=1,\dots , M$. The maps $\Pi_N A_j, \sigma$ and $\ell_j$ are finite dimensional when restricted to $\Span\{e_1,\dots ,e_N\}$ and the sum above thus resembles a classical neural network. However, instead of the typical one dimensional activation function, the function $\Pi_N \circ \sigma$ restricted to $\Span\{e_1,\dots ,e_N\}$ is multidimensional. 
\end{remark}

We notice that the basis functions $\{e_i\}_{i\in\mathbb N}$ incorporate structural information about the objects in the Hilbert space. In applications, this may provide an advantage in training a neural net in Hilbert space rather than using classical sampling, and, moreover, it provides a net that is trained as functions on the basis expansion rather than on sampling the input functions. Furthermore, if the activation  function $\sigma$ is (Fr\'echet) differentiable, the neural network will be (Fr\'echet) differentiable, and we may derive analytic expressions for the sensitivities in directions along basis functions, say, avoiding numerical differentiation.

\section{Modeling the flow forward price dynamics}\label{sec:flow-forwards}

Let $(\Omega,\mathcal F,(\mathcal F_t),\mathbb Q)$ be a  filtered probability space satisfying the {\it usual conditions} (see Definition 2.25 in Karatzas and Shreve \cite{KS}). In our analysis of option prices, we assume right away that we work under the risk neutral probability $\mathbb Q$. 

Furthermore, as state space of our stochastic models for the flow forwards, let $H$ be a separable Hilbert space of measurable functions $f:\mathbb R_+\rightarrow\mathbb R$, with inner product $\langle\cdot,\cdot\rangle$ and induced norm $\vert\cdot\vert$.
We assume that the evaluation functionals
$\delta_{\xi}\in H^*$ for any $\xi\in\mathbb R_+$, with $\delta_{\xi}(f):=f(\xi)$, and that the family of (right-)shift operators
$(\mathcal S_t)_{t\in\mathbb R_+}$, with $\mathcal S_t f:=f(t+\cdot)$ forms a $C_0$-semigroup on $H$ which is quasi-contractive. Its densely defined generator is the derivative operator $\partial_{\xi}$.  Finally, $H$ is supposed to be a Banach algebra under pointwise multiplication of functions where the constant unit function $\xi\mapsto 1\in H$. We include the important example of the Filipovi\'c space next:  

\begin{example}\label{exp:Filipovic:space}
An example of a separable Hilbert space of measurable real-valued functions on $\mathbb R_+$ is given by the Filipovi\'c space, which was introduced in Filipovi\'c \cite{Filip}.
Denote by $H_w$ the space of absolutely continuous functions $x:\mathbb R_+\rightarrow \mathbb R$ satisfying the norm
$$
\vert x\vert_w^2=x(0)^2+\int_0^{\infty}w(\xi)x'(\xi)^2d\xi<\infty.
$$
The weight function $w:\mathbb R_+\rightarrow[1,\infty)$ is non-decreasing and measurable with $w(0)=1$. Under integrability of $w^{-1}$, the shift semigroup $(\mathcal S_t)_{t\geq 0}$ associated to the densely defined generator $\partial_{\xi}$ is a $C_0$-semigroup. Moreover, $H_w$ is a Banach algebra with respect to pointwise multiplication (see Prop. 4.18 in Benth and Kr\"uhner \cite{BK-CIMS}) supporting constant functions. Furthermore, the shift-semigroup is quasi-contractive (see Prop. 4.4 in Benth, Detering and Kr\"uhner \cite{BDK-stochastics}). 
\end{example}

\subsection{Stochastic modeling of fixed-delivery forwards}

Assume that $(X(t))_{t\geq 0}$ satisfies the stochastic partial differential equation (SPDE) 
\begin{equation}\label{SPDE}
    dX(t)=\partial_{\xi} X(t)dt+\alpha(t,X(t))dt+\eta(t,X(t))dW(t)+\int_{H}\gamma(t,X(t),z)\widetilde{N}(dt,dz)
\end{equation}
with $X(0):=X_0$ for an $\mathcal F_0$-measurable $H$-valued square integrable random variable. Here $W$ is a Wiener process in $H$ with a positive definite trace class (covariance) operator $\mathcal Q$, and 
$\widetilde{N}(dt,dz):=N(dt,dz)-dt\otimes \nu(dz)$ for a homogeneous Poisson random measure $N(dt,dz)$ on $\mathbb R_+\times H$ with compensator $dt\otimes\nu(dz)$. Here, $\nu$ is a $\sigma$-finite L\'evy measure on $H$. Furthermore, we have measurable coefficients $\alpha:\mathbb R_+\times H\rightarrow H$, $\eta:\mathbb R_+\times H\rightarrow L_{\text{HS}}(H)$ and
$\gamma:\mathbb R_+\times H\times H\rightarrow H$, where $L_{\text{HS}}(H)$ is the space of Hilbert-Schmidt operators on $H$.  

Following Filipovi\'c, Tappe and Teichmann \cite{FTT}, one can show under rather mild conditions on the coefficient functions that there exists a unique mild solution of the SPDE \eqref{SPDE}. We describe this in more detail: for $p\in\mathbb N$, let $L^p_{loc}(\mathbb R_+)$ be the space of locally $p$-integrable real-valued functions on $\mathbb R_+$. Assume that 
\begin{align*}
    \vert\alpha(\cdot,0)\vert&\in L^2_{loc}(\mathbb R_+) \\
    \Vert\eta(\cdot,0)\Vert_{\text{HS}}&\in L^2_{loc}(\mathbb R_+) \\
    \int_H\vert\gamma(\cdot,0,z)\vert^2\nu(dz)&\in L^1_{loc}(\mathbb R_+),
\end{align*}
where $\Vert\cdot\Vert_{\text{HS}}$ denotes the Hilbert-Schmidt norm. Further, there exists a function $C\in L^2_{loc}(\mathbb R_+)$  such that the following Lipschitz continuity holds,
\begin{align*}
    \vert\alpha(t,x)-\alpha(t,y)\vert&\leq C(t)\vert x-y\vert \\
    \Vert\eta(t,x)-\eta(t,y)\Vert_{\text{HS}}&\leq C(t)\vert x-y\vert \\
    \int_H\vert\gamma(t,x,z)-\gamma(t,y,z)\vert^2\nu(dz)&\leq C(t)^2\vert x-y\vert^2
\end{align*}
Under these conditions we recall Corollary 10.6 from Filipovi\'c, Tappe and Teichmann \cite{FTT}:
\begin{theorem}
\label{thm:ex-unique-X}
There exists a unique adapted mean-square continuous $H$-valued stochastic process $(X(t))_{t\geq 0}$ which is right-continuous with left-limits (RCLL) being a mild solution of \eqref{SPDE}, i.e., solving the integral equation
\begin{align}
\label{log:energy:price}
X(t)&=\mathcal S_t X_0+\int_0^t\mathcal S_{t-s}\alpha(s,X(s))ds+\int_0^t\mathcal S_{t-s}\eta(s,X(s))dW(s) \nonumber\\
&\qquad+\int_0^t\int_H\mathcal S_{t-s}\gamma(s,X(s-),z)\widetilde{N}(ds,dz),
\end{align}
where $X(s-)=\lim_{u\uparrow s}X(u)$. Moreover, 
$$
\mathbb E\left[\sup_{0\leq t\leq\tau}\vert X(t)\vert^2\right]<\infty
$$
for all $\tau>0$. 
\end{theorem}

\begin{remark}
Notice that in Filipovi\'c, Tappe and Teichmann \cite{FTT}, the class of SPDEs analysed is more general than \eqref{SPDE}. The Hilbert space $H$ is not restricted to real-valued functions on $\mathbb R_+$. Furthermore, they treat general unbounded operators $A$ rather than merely $A=\partial_{\xi}$ as we focus on here. The operators $A$ must generate a $C_0$ quasi-contractive semigroup in their context. Next, the Wiener process $W$ can take values in a different Hilbert space than $H$, with the natural modification that $\eta$ is a Hilbert-Schmidt operator mapping from this space into $H$. Finally, they let the mark space of the Poisson random measure be a Blackwell space (again with appropriate modification of the function $\gamma$). Also, we remark in passing that Tappe \cite{Tappe} has relaxed the conditions on the parameter functions to be of local Lipschitz continuity and linear growth. We focus in the present paper on a special case to tailormake the situation to our application to financial forward contracts. 
\end{remark}
We note that in the case $\alpha, \eta$ and $\gamma$ are not state dependent, we demand only local integrability on the coefficient functions as Lipschitz continuity is trivially fulfilled.  We denote the mild solution $Y$ (and the initial condition $Y_0$ being $H$-valued $\mathcal F_0$-measurable random variable) in this case, which is explicitly given by
\begin{equation}
\label{eq:explicit-price-non-state-dep}
Y(t)=\mathcal S_tY_0+\int_0^t\mathcal S_{t-s}\alpha(s)ds+\int_0^t\mathcal S_{t-s}\eta(s)dW(s)+\int_0^t\int_H\mathcal S_{t-s}\gamma(s,z)\widetilde{N}(ds,dz).
\end{equation}
We emphasise that $Y$ is not a solution of an integral 
equation, but explicitly given in \eqref{eq:explicit-price-non-state-dep} in terms of (stochastic) integrals of the coefficient functions.

We define the forward price dynamics under the Musiela parametrisation as follows: with $F(t,T)$ being the forward price at time $t\geq 0$ for a contract delivering the underlying commodity at time $T\geq t$, let 
\begin{equation}
X(t,\xi):=F(t,t+\xi)
\end{equation}
where $\xi:=T-t\in\mathbb R_+$. Or, we may express the forward price in terms of the evaluation functional applied to $X$,
\begin{equation}
\label{eq:forward-HJM-markov}
    F(t,T)=\delta_{T-t}X(t)=X(t,T-t).
\end{equation}
From arbitrage theory (see e.g. Duffie \cite{Duffie}), the process $t\mapsto \hat{F}(t,T_1,T_2)$ for $t\leq T_1$ must be a (local) martingale to ensure no-arbitrage in the market. To ensure this, it is convenient to assume
that $t\mapsto F(t,T)$ for $t\leq T$ is a (local) martingale. It is readily seen that $\mu=0$ is a sufficient condition guaranteeing this, and we assume so in the sequel.  
 
 One may also consider an alternative geometric model based on the situation where the coefficient functions are not state-dependent. Supposing that $Y$ is given by \eqref{eq:explicit-price-non-state-dep}, we define $t\mapsto X(t):=\exp(Y(t))$ as the price dynamics. Recall that $H$ is a Banach algebra, and thus $\exp(Y(t))\in H$. We have that 
\begin{equation}
\label{eq:forward-HJM-geometric}
F(t,T)=\delta_{T-t}X(t)=\delta_{T-t}\exp(Y(t))=\exp(Y(t,T-t)).
\end{equation}
Thus, we ensure by the exponential modelling that the forward prices $F(t,T)$ are positive. Also, appealing to the algebra structure of $H$, we can express $X(t)$ as
\begin{equation}
\label{eq:markovian-struct-exponential}
    X(t)=\exp(Y(t))=\exp(\mathcal S_t Y_0)\exp(Y^0(t))=(\mathcal S_t X_0)\exp(Y^0(t))
\end{equation}
where $Y^0(t)$ is $Y(t)$ with initial condition being zero.  
To have an arbitrage-free forward price dynamics $t\mapsto F(t,T), t\leq T$, we must in this case demand that (see Benth and Kr\"uhner \cite[Prop. 6.3]{BK-book}) 
\begin{equation}
\label{arb:free:drift}
\alpha(t,\cdot):=-\frac12\vert \mathcal Q^{1/2}\eta^*(t)(\delta_{\cdot}^*1)\vert^2-\int_H\left(\exp(\gamma(t,z))-1-\gamma(t,z)\right)\nu(dz) 
\end{equation}
and that $\alpha(t)\in H$
for the non-state dependent case. It is also worth noticing that the no-arbitrage condition relies on the finite exponential moment integrability condition of $\nu$, which yields that $\exp(Y^0(\tau))$ is integrable.   
\begin{example}
Recall the Filipovi\'c space $H_w$ introduced in Example \ref{exp:Filipovic:space}.
Indeed, it is possible to show that $\alpha(t)\in H_w$ (see Prop. 6.3 and the following discussion in Benth and Kr\"uhner \cite{BK-book}). Here, notice that $h_{\xi}(u):=\delta_{\xi}^*1(u)$ is the representative of the linear functional $\delta_{\xi}$ in $H_w$, meaning that $\delta_{\xi} x=\langle h_{\xi},x\rangle_w$. This function is explicitly given in $H_w$ as
$$
h_{\xi}(u)=1+\int_0^{\xi\wedge u}w^{-1}(v)dv
$$
A natural specification of $\sigma$ could be simply a model of the Samuelson effect, i.e., $\eta(t)x(\xi):=a\exp(-b\xi)x(\xi)$. Hence, we specify the operator $\eta(t)$ to be constant in time, but a simple multiplication operator on $H_w$ by the function $\xi\mapsto a\exp(-b\xi)$, (with the constants $a,b>0$) which is an element of $H_w$ under some natural conditions of the scale $w$ of $H_w$ (see Thm. 4.17 in Benth and Kr\"uhner \cite{BK-CIMS}). 
\end{example}

Before continuing, we emphasise that in the sequel of this paper we shall consider two forward price models, one being \eqref{eq:forward-HJM-markov} based on $X$ solving the SPDE \eqref{SPDE} under the condition that $\mu=0$, or the geometric model \eqref{eq:forward-HJM-geometric} where $Y$ is the dynamics in \eqref{eq:explicit-price-non-state-dep} satisfying the no-arbitrage condition in \eqref{arb:free:drift}.

\subsection{Stochastic modelling of flow forwards}

Let us now consider a flow forward with delivery over the time interval $[T_1,T_2]$. Recalling the flow forward price at time $t\leq T_1$ by $\widehat{F}(t,T_1,T_2)$, where, 
\begin{equation}
\label{def:flow-forward}
    \widehat{F}(t,T_1,T_2)=\frac1{T_2-T_1}\int_{T_1}^{T_2}F(t,T)dT.
\end{equation}
Writing $\xi:=T_1-t$ and $\lambda:=T_2-T_1$, we find 
\begin{equation}
\label{eq:integral-flow-forward}
\widehat{X}_{\lambda}(t,\xi):=\widehat{F}(t,t+\xi,t+\xi+\lambda)=\frac1{\lambda}\int_0^{\infty}\mathrm{1}_{[0,\lambda]}(u-\xi)X(t,u)du.
\end{equation}
Here, $X$ is the process for the fixed-delivery forward in the previous Subsection. Hence, the flow forward price can be defined as a stochastic process $\widehat{X}_{\lambda}(t)$ taking values in a space of measurable real-valued functions on $\mathbb R_+$, $\xi\mapsto\widehat{X}_{\lambda}(t,\xi)$. 
Indeed, we can write $\hat{X}_{\lambda}(t)=\mathcal D_{\lambda}(X(t))$, for an integral operator defined from \eqref{eq:integral-flow-forward} by
\begin{equation}
\mathcal D_{\lambda}(f)=\frac{1}{\lambda}\int_0^{\infty}\mathrm{1}_{[0,\lambda]}(u-\cdot)f(u)du,
\end{equation} 
with $f\in H$. 
From now on, we assume that $\mathcal D_{\lambda}\in L(H)$, i.e., a bounded linear operator on $H$, such that $\widehat{X}_{\lambda}(t)$ becomes an $H$-valued stochastic process. 

\begin{remark}
In Benth and Kr\"uhner \cite[Prop. 2.1.]{BK-SIFIN}, it is shown that $\mathcal D_{\lambda}\in L(H_w)$. In fact, this holds for more general flow forward integral operators associated with relations with the fixed-delivery forward beyond \eqref{def:flow-forward}.
\end{remark}

\section{The option price functional}
\label{sec:option}

Consider an option written on a flow forward contract delivering over $[T_1,T_2]$, where we assume that the option has exercise time $\tau\leq T_1$ with payoff defined by a measurable function $\mathfrak{P}:\mathbb R\rightarrow\mathbb R$. Assuming that 
$\mathfrak{P}(\widehat{F}(\tau,T_1,T_2))$ is integrable, the no-arbitrage price at time $t\leq\tau$ of the option is given by (see e.g. Bj\"ork \cite{bjork})
$$
V(t):=e^{-r(\tau-t)}\mathbb E\left[\mathfrak{P}(\widehat{F}(\tau,T_1,T_2))\,\vert\,\mathcal F_t\right].
$$
Here and in the sequel we assume for simplicity that the risk-free interest rate $r>0$ is constant and deterministic. 

Using the definitions in the previous Section, the option payoff can be written
\begin{align*}
\mathfrak{P}(\widehat{F}(\tau,T_1,T_2))&=\mathfrak{P}(\widehat{F}(\tau,\tau+(T_1-\tau),\tau+(T_1-\tau)+T_2-T_1)) \\
&=\mathfrak{P}(\delta_{T_1-\tau}\mathcal D_{T_2-T_1}(X(\tau))).
\end{align*}
Hence, the price becomes
\begin{equation}
    V(t)=e^{-r(\tau-t)}\mathbb E\left[\mathfrak{P}\left(\delta_{T_1-\tau}\mathcal D_{T_2-T_1}(X(\tau)\right)\,\vert\,\mathcal F_t\right].
\end{equation}

We have the following:
\begin{proposition}
Let $X$ be given by the dynamics in \eqref{SPDE}, and assume that the coefficient functions $\alpha$, $\eta$ and $\gamma$ are independent of time (i.e., only state dependent). It then holds that $V(t):=V(t,X(t))$ where
\begin{equation}
\label{eq:price-markovian}
    V(t,x)=e^{-r(\tau-t)}\mathbb E\left[\mathfrak{P}\left(\delta_{T_1-\tau}\mathcal D_{T_2-T_1}(X^{t,x}(\tau))\right)\right].
\end{equation}
Here, $X^{t,x}(\tau)$ means that $X^{t,x}(t)=x\in H$.
\end{proposition}
\begin{proof}
This follows from the Markovian property of $X$ in the case of coefficients being independent of time, see e.g., Th. 9.30 in Peszat and Zabczyk \cite{PZ}. 
\end{proof}
I.e., we have that the price of the option becomes a functional of $X(t)$. 
Recall \eqref{eq:markovian-struct-exponential} for the exponential price defined in \eqref{eq:forward-HJM-geometric}. Hence, from \eqref{eq:explicit-price-non-state-dep} we find for $Y(\tau)$ given $Y(t)$, $\tau\geq t$,
$$
Y(\tau)=\mathcal S_{\tau-t}Y(t)+\int_t^{\tau}\mathcal S_{\tau-s}\alpha(s)ds+\int_t^{\tau}\mathcal S_{\tau-s}\eta(s)dW(s)+\int_t^{\tau}\int_H\gamma(s,z)\widetilde{N}(ds,dz).
$$
Therefore, 
$$
X(\tau)=\exp(\mathcal S_{\tau-t}Y(t))\exp(Z_{t,\tau})
$$
where $Z_{t,\tau}$ is defined as
\begin{equation}
\label{eq:Z-definition}
Z_{t,\tau}:=\int_t^{\tau}\mathcal S_{\tau-s}\alpha(s)ds+\int_t^{\tau}\mathcal S_{\tau-s}\eta(s)dW(s)+\int_t^{\tau}\int_H\gamma(s,z)\widetilde{N}(ds,dz)
\end{equation}
with the no-arbitrage condition \eqref{arb:free:drift} in place for $\alpha$. Thus, $Z_{t,\tau}$
is a random variable in $H$ independent of $\mathcal F_t$ by the independent increment property of the Wiener process and the jump measure. Furthermore, the no-arbitrage condition \eqref{arb:free:drift} involves an exponential integrability condition on the L\'evy measure, which, together with Fernique's Theorem implies that $\exp(Z_{t,\tau})$ is integrable. Moreover, by using the shift semigroup, 
$$
\exp(\mathcal S_{\tau-t}Y(t))=\mathcal S_{\tau-t}\exp(Y(t))=\mathcal S_{\tau-t}X(t)
$$
Thus, also in this case we find $V(t)=V(t,X(t))$ where
\begin{equation}
\label{eq:price-markovian-exponential}
V(t,x)=e^{-r(\tau-t)}\mathbb E\left[\mathfrak{P}\left(\delta_{T_1-\tau}\mathcal D_{T_2-T_1}((\mathcal S_{\tau-t}x)\exp(Z_{t,\tau}))\right)\right].
\end{equation}
\begin{remark} 
Notice that the $X(t)$ is a positive-valued function in our Hilbert space as it is defined by $X(t)=\exp(Y(t))$. Thus, $V(t,x)$ in \eqref{eq:price-markovian-exponential} is defined on the domain $H_+$, that is, the subset of $H$ consisting of positive-valued functions in $H$, when we assume that this is the price of the option. However, looking at the actual expression in \eqref{eq:price-markovian-exponential}, it obviously makes sense as it stands for all $x\in H$. We shall understand $V(t,\cdot)$ as a function defined on all of $H$.
\end{remark}
Hence, for both our dynamical models of the flow forward price, the corresponding option price is a functional $V:[0,\tau]\times H\rightarrow \mathbb R$.

The following is an important motivation for the studies in this paper:
Notice that for the geometric model with $Z_{t,\tau}$ given in \eqref{eq:Z-definition}
\begin{align*}
   \delta_{T_1-\tau}\mathcal D_{T_2-T_1}((\mathcal S_{\tau-t}x)\exp(Z_{t,\tau}))&=\frac1{T_2-T_1}\int_{T_1-\tau}^{T_2-\tau}x(u+\tau-t)e^{Z_{t,\tau}(u)}du \\
   &=\frac1{T_2-T_1}\int_{T_1-t}^{T_2-t}x(v)e^{Z_{t,\tau}(v-(\tau-t))}dv.
\end{align*}
From this it is evident that the option price does not depend on $\widehat{F}(t,T_1,T_2)$ only, but rather the whole curve $\xi\mapsto X(t,\xi)$, that is, the term structure of fixed-forward prices. I.e., the option price is a functional $V(t,\cdot):H\rightarrow\mathbb R$ and cannot be reduced to a function depending only on $\widehat{F}(t,T_1,T_2)$.

We continue analysing the price functional $V$ in \eqref{eq:price-markovian} and \eqref{eq:price-markovian-exponential}, first showing the crucial property of Lipschitz continuity of $x\mapsto V(t,x)$. This is fundamental in applying a feedforward neural network to compute $V$.  

\begin{proposition}
\label{prop:lipschitz}
Assume $\mathfrak{P}$ is Lipschitz continuous. Then $V$ is well-defined and Lipschitz continuous, i.e., for any $x,y\in H$ and $t\in[0,\tau]$,
$$
\vert V(t,x)-V(t,y)\vert\leq K\vert x-y\vert 
$$
for some constant $K>0$.
\end{proposition}
\begin{proof}
As $\mathfrak{P}$ is Lipschitz continuous, 
$$
\vert\mathfrak{P}(u)\vert\leq K\vert u\vert+\vert \mathfrak{P}(0)\vert
$$
for some constant $K>0$. Therefore, from the boundedness of the operators $\delta_{\xi}$, $\mathcal D_{\lambda}$ 
$$
\vert\mathfrak{P}(\delta_{T_1-\tau}\mathcal D_{T_2-T_1}(X(\tau)))\vert\leq K\Vert\delta_{T_1-\tau}\Vert_{\text{op}}\Vert\mathcal D_{T_2-T_1}\Vert_{\text{op}}\vert X(\tau)\vert+\vert\mathfrak{P}(0)\vert.
$$
From the uniform $L^2$-bound on $X(\tau)$ in Thm. \ref{thm:ex-unique-X} the well-definedness of $V$ follows when $X$ is given by \eqref{SPDE}. When $X(\tau)=\exp(Y(\tau))$ with $Y$ as in
\eqref{eq:explicit-price-non-state-dep}, then from \eqref{eq:markovian-struct-exponential} and the definitions of the operators $\delta_{\xi}$ and $\mathcal D_{\lambda}$,
$$
\delta_{T_1-\tau}\mathcal D_{T_2-T_1}(X(\tau))=\frac{1}{T_2-T_1}\int_0^{\infty}\mathrm{1}_{[0,T_2-T_1]}(u-(T_1-\tau))\exp(y(\tau+u))\exp(Y^0(\tau,u))du.
$$
By the no-arbitrage drift condition 
$\mathbb E[\exp(Y^0(\tau,u))]=1$, we have by Tonelli's theorem
$$
\mathbb E[\delta_{T_1-\tau}\mathcal D_{T_2-T_1}(X(\tau))]=\delta_{T_1-\tau}\mathcal{D}_{T_2-T_1}(\exp(\mathcal S_{\tau}y))
$$
which is finite as $\mathcal D_{\lambda}\in L(H)$. Thus, from the linear growth of $\mathfrak{P}$ the well-definedness of $V$ also follows for the exponential model. 

To show Lipschitz continuity, consider first the case of $V$ in the exponential model \eqref{eq:price-markovian-exponential}. By Lipschitz continuity of $\mathfrak{P}$, 
\begin{align*}
\vert V(t,x)&-V(t,y)\vert \\
&\leq K\mathbb E[\vert\delta_{T_1-\tau}\mathcal D_{T_2-T_1}(\mathcal S_{\tau-t}(x-y)e^{Z_{t,\tau}})\vert] \\
&\leq K\Vert\delta_{T_1-\tau}\Vert_{\text{op}}\Vert\mathcal D_{T_2-T_1}\Vert_{\text{op}}\Vert\mathcal S_{\tau-t}\Vert_{\text{op}}\mathbb E[\vert\exp(\vert Z_{t,\tau})\vert]\vert x-y\vert.
\end{align*}
Again, as all the involved operators are bounded and $Z_{t,\tau}$ is an exponentially integrable random variable, the assertion holds. In the Markovian model, we appeal again to the Lipschitz continuity of $\mathfrak{P}$, and using the Cauchy-Schwarz inequality we reach
$$
\vert V(t,x)-V(t,y)\vert\leq K\Vert\delta_{T_1-\tau}\Vert_{\text{op}}\Vert\mathcal D_{T_2-T_1}\Vert_{\text{op}}\mathbb E\left[\vert X^{t,x}(\tau)-X^{t,y}(\tau)\vert^2\right]^{1/2}.
$$
By Thm. 9.29 (part (1.ii)) in Peszat and Zabczyk \cite{PZ}, the latter expectation is bounded by $\vert x-y\vert$.  
Thus, we find the desired Lipschitz continuity of $V$. 
\end{proof}
\begin{remark} When pricing an option on a given flow forward, one may ask the question why model the whole curve and not simply only $t\mapsto \widehat{F}(t,T_1,T_2)$ directly and then price the option? An argument for modeling the whole term-structure curve is that we may have many forward contracts and different options on these. It is desirable that all the option prices are internally consistent, in the sense that one cannot obtain arbitrage. This will be guaranteed when modelling the whole forward curve dynamics. Of course, one could model all the (finite number) of forwards as a multi-dimensional process, but then one looses the possibility to also include other forwards, possibly not listed in the market. For the latter case, imagine a trader getting an offer of an option on a new forward not part of her multivariate model. Then the whole model must be re-defined, to include the new forward into the pricing. 
\end{remark}
\begin{remark}
Considering $V$ as in \eqref{eq:price-markovian-exponential}, it holds that $H\ni y\mapsto V(t,e^y)$ is locally Lipschitz continuous, since the map $y\mapsto e^y$ is locally Lipschitz continuous. Notice that with the exponential model, we can consider the option price $V(t,X(t))$ as a functional of $Y(t)$, i.e., $V(t)=V(t,\exp(Y(t))$.
\end{remark}

\section{Approximating the price functional by neural networks}\label{sec:approx:price:fct}
To prepare for approximating the option price $V(t,x)$ by a neural network, we recast it as a solution of an optimization problem following Beck et al. \cite{Beck2021}. First, we notice that both in the Markovian and geometric forward price models, the option price functional may be expressed generically as a functional $\overline{V}: H\rightarrow \mathbb R$,
\begin{equation}
\overline{V}(x)=\mathbb E[\mathcal X(x)]
\end{equation}
where $\mathcal X$ is a random field on $H$. Indeed, from \eqref{eq:price-markovian-exponential} we find for the geometric model that 
\begin{equation}
\label{eq:X-def-geom}
\mathcal X(x)=e^{-r(\tau-t)}\mathfrak{P}(\delta_{T_1-\tau}\mathcal D_{T_2-T_1}((\mathcal S_{\tau-t}x)\exp(Z_{t,\tau})))
\end{equation}
with $Z_{t,\tau}$ in \eqref{eq:Z-definition}, and for the Markovian model we derive from \eqref{eq:price-markovian}
\begin{equation}
\label{eq:X-def-markov}
\mathcal X(x)=e^{-r(\tau-t)}\mathfrak{P}(\delta_{T_1-\tau}\mathcal D_{T_2-T_1}(X^{t,x}(\tau))
\end{equation}
Notice that we have ignored the dependency on $t$, $T_1$, $T_2$ and $\tau$ here, as these do not play a role in what follows. 
We recall that $\overline{V}$ is a Lipschitz continuous mapping on $H$ (see Prop. \ref{prop:lipschitz}).

To this end, 
introduce a measure $\mu:\mathcal B(H)\rightarrow[0,\infty]$ and assume 
\begin{equation}
\label{cond:q-integration}
\mathbb E\left[\int_{H}\mathcal X^2(x)\mu(dx)\right]<\infty.
\end{equation}
This assumption is a joint condition on the the measure $\mu$ and the random field $\mathcal X(x)$. In the next Lemma, we state a sufficient condition on $\mu$:
\begin{lemma}
\label{lemma:suff-mu-int}
Assume for the exponential model that $\exp(2\vert Z_{t,\tau}\vert)$ is integrable, with $Z_{t,\tau}$ defined in \eqref{eq:Z-definition}.
If 
$$
\int_H\max(1,\vert x\vert^2)\mu(dx)<\infty
$$
then condition \eqref{cond:q-integration} holds.
\end{lemma}
\begin{proof}
For simplicity, we suppose that $r=0$, and recall the Lipschitz continuity of $\mathfrak{P}$ resulting in the linear growth bound
$$
\vert\mathfrak{P}(z)\vert^2\leq C_1+C_2z^2
$$
for positive constants $C_1, C_2$. 

Consider first the Markovian case of $\mathcal X(x)$ in \eqref{eq:X-def-markov}: using $\mathcal L:=\delta_{T_1-\tau}\mathcal D_{T_2-T_1}$ as short-hand notation, we have
$$
\mathbb E[\mathcal X^2(x)]=\mathbb E[\mathfrak{P}^2(\mathcal L(X^{t,x}(\tau))]\leq C_1+C_2\Vert\mathcal L\Vert_{\text{op}}^2\mathbb E[\vert X^{t,x}(\tau)\vert^2].
$$
Thm. 9.29(ii) of Peszat and Zabczyk \cite{PZ} along with the uniform $L^2$-bound in Thm. \ref{thm:ex-unique-X} yield, after using the triangle inequality,
\begin{align*}
  \mathbb E[\vert X^{t,x}(\tau)\vert^2]&\leq 2\mathbb E[\vert X^{t,0}(\tau)\vert^2]+2\mathbb E[\vert X^{t,x}(\tau)-X^{t,0}(\tau)\vert^2]  \\
  &\leq 2C(1+\vert x\vert^2),
\end{align*}
for some positive constant $C$. Hence, by Tonelli's Theorem (and with $C$ now denoting a generic constant)
$$
\mathbb E[\int_H\mathcal X^2(x)\mu(dx)]=\int_H\mathbb E[\mathcal X^2(x)]\mu(dx)\leq C\int_H1+\vert x\vert^2\mu(dx).
$$
Thus the claim follows for the Markovian case. 

Next, consider the exponential model for $\mathcal X(x)$ in \eqref{eq:X-def-geom}. Then it follows that
$$
\mathbb E[\mathcal X^2(x)]=\mathbb E[\mathfrak{P}^2(\mathcal L((\mathcal S_{\tau-t}x))\exp(Z_t,\tau))]\leq C_1+C_2\Vert\mathcal L\Vert_{\text{op}}^2\Vert\mathcal S_{\tau-t}\Vert_{\text{op}}^2\vert x\vert^2\mathbb E[\exp(2\vert Z_{t,\tau}\vert)].
$$
By the integrability assumption on $Z_{t,\tau}$, the claim follows. 
\end{proof}
In the case of an exponential model without jumps, Fernique's Theorem (see e.g. Thm. 3.31 in Peszat and Zabczyk \cite{PZ}) ensures that $\exp(2\vert Z_{t,\tau}\vert)$ is integrable. This is readily seen from the elementary inequality $2a b\leq a^2\epsilon^{-2}+\epsilon^2b^2$ for any $\epsilon>0$, yielding
$$
2\vert Z_{t,\tau}\vert\leq \epsilon^{-2}+\epsilon^2\vert Z_{t,\tau}\vert^2.
$$
If we have jumps in the model, the exponential integrability of $2\vert Z_{t,\tau}\vert$ can be translated into an exponential integrability condition on the L\'evy measure.  

In the application to neural nets, it is the case of $\mu$ having compact support which is of interest to us. We see from the Lemma above that under mild additional integrability hypotheses for the exponential model in the jump case, condition \eqref{cond:q-integration} is satisfied in this case as $x\mapsto\max(1,\vert x\vert^2)$ is a continuous function on $H$.

Consider the class of functions $g\in CL^2_{lip}(\mu)$, where $CL^2_{lip}(\mu)$ is the set of real-valued Lipschitz continuous functions on $H$ where $\int_Hg^2(x)\mu(dx)<\infty$. Obviously, $CL^2_{lip}(\mu)\subset L^2(\mu)$. Furthermore, under the conditions of Lemma \ref{lemma:suff-mu-int}, it follows that $\overline{V}\in CL_{lip}^2(\mu)$. The following Lemma shows that $\overline{V}$ is the global minimizer on $L^2(\mu)$: 
\begin{lemma}
It holds that 
\begin{equation}
\overline{V}(\cdot)=\arg\min_{g\in L^2(\mu)}\mathbb E\left[\int_H\vert \mathcal X(x)-g(x)\vert^2\mu(dx)\right].
\end{equation}
\end{lemma}
\begin{proof}
Suppose $g\in L^2(\mu)$: a direct calculation shows that
\begin{align*}
    &\mathbb E\left[\int_H\vert \mathcal X(x)-g(x)\vert^2\mu(dx)\right] \\
    &\qquad=\mathbb E\left[\int_H \mathcal X^2(x)\mu(dx)\right]-2\mathbb E\left[\int_H \mathcal X(x)g(x)\mu(dx)\right]+\int_H g^2(x)\mu(dx) \\
    &\qquad=\mathbb E\left[\int_H \mathcal X^2(x)\mu(dx)\right]-2\int_H \overline{V}(x)g(x)\mu(dx)+\int_H g^2(x)\mu(dx) \\
    &\qquad=\int_H\mathbb E[(\mathcal X^2(x)-\mathbb E[\mathcal X(x)]^2)]\mu(dx)+\int_H\vert \overline{V}(x)-g(x)\vert^2\mu(dx) \\
    &\qquad=\int_H\text{Var}(\mathcal X(x))\mu(dx)+\int_H\vert \overline{V}(x)-g(x)\vert^2\mu(dx).
\end{align*}
The result follows.
\end{proof}
In Prop. 2.2 of Beck et al. \cite{Beck2021}, a finite-dimensional analogue  of the above Lemma is shown. They use this as the motivation for approximating the minimizer by neural networks.  
Below we provide a rigorous argument for why one can minimize over the subset of "all" neural networks rather than $L^2(\mu)$:

First of all, we recall the following easy fact. If $(\mathcal T,\tau)$ is a topological space, $\mathcal U\subset \mathcal T$ such that $cl(\mathcal U)\cap \mathcal T =\mathcal T$ (i.e. $\mathcal U$ is dense in $\mathcal T$), and $I:\mathcal T\to\R$ is continuous, then 
$$
\inf_{\mathcal T} I = \inf_{\mathcal U} I \in [-\infty,\infty).
$$
Let us embed our current setup in this framework. 
\begin{lemma}
\label{Lemma:I-Lip}
Let $\mu:\mathcal{B}(H)\to [0,\infty]$ be a measure. Then the map 
$$
I:L^2(\mu)\to \R,\quad g \mapsto \mathbb E\left[\int_H\vert \mathcal X(x)-g(x)\vert^2\mu(dx)\right] 
$$
is locally Lipschitz continuous
\end{lemma}
\begin{proof}
Denote by $\Vert\cdot\Vert_{\mu}$ the norm in $L^2(\mu)$. Observe that 
$$
\vert a-b\vert^2-\vert a-c\vert^2=2a(c-b)+(b+c)(b-c)
$$
for any constants $a,b,c\in\mathbb R$. Hence, for
$a=\mathcal X(x)$, $b=g(x)$ and $c=h(x)$ for $g,h\in L^2(\mu)$, we find 
\begin{align*}
\vert I(g)-I(h)\vert&\leq2\mathbb E\left[\int_H\vert \mathcal X(x)\vert\vert g(x)-h(x)\vert\mu(dx)\right]+\int_H\vert g(x)+h(x)\vert\vert g(x)-h(x)\vert\mu(dx) \\
&\leq 2\mathbb E\left[\int_H\mathcal X^2(x)\mu(dx)\right]^{1/2}\int_H\vert g(x)-h(x)\vert^2\mu(dx)^{1/2} \\
&\qquad+\int_H\vert g(x)+h(x)\vert^2\mu(dx)^{1/2}\int_H\vert g(x)-h(x)\vert^2\mu(dx)^{1/2}
\end{align*}
after using the Cauchy-Schwarz inequality. The claim follows by assumption \eqref{cond:q-integration}. 
\end{proof}

Therefore, a fortiori, the map $I$ defined in Lemma \ref{Lemma:I-Lip} above is continuous on $CL^2_{lip}(\mu)$. We now fix a compact subset $K\subset H$, and assume that $\Supp \mu = K$ and $\mu(K)=1$. 
Recalling our standard assumptions on the neural network from Section \ref{nn:frechet} (in particular assuming that the activation function $\sigma$ is Lipschitz, so that the resulting neural networks become Lipschitz), our abstract result Prop. \ref{prop: density} ensures that 
$$
cl(\mathfrak{N}(\sigma)) \cap CL^2_{lip}(\mu) = CL^2_{lip}(\mu)
$$
with respect to the topology of convergence on compacts. Since $\mu(K)<\infty$, it follows that this last equation holds good even with respect to the $L^2(\mu)$-topology. Therefore,
$$
I(\overline{V}) =  \min_{L^2(\mu)} I  = \min_{ CL^2_{lip}(\mu)} I = \inf_{ CL^2_{lip}(\mu)} I = \inf_{\mathfrak{N}(\sigma)} I. 
$$
Hence, we have justified the procedure in finite dimensions of Beck et al. \cite{Beck2021} for approximating the minimizer by neural networks. And even more, we have demonstrated that $\bar{V}$ can be approximated arbitrary well by our infinite-dimensional  neural networks.

Observe that if we consider the function $\R\ni \varepsilon\mapsto I(g+\varepsilon v)$ with $g,v\in L^2(\mu)$ and we compute its second derivative we get $2\norm{v}^2_{L^2(\mu)}$, showing that $I$ is (strictly) convex on $L^2(\mu)$.

It might be worth noticing that for each $g\in L^2(\mu)$ it holds in view of H\"older's inequality
\begin{equation*}
    \begin{split}
        I(g) &= \mathbb E\left[\int_H\mathcal X^2(x)\mu(dx)\right]-2\mathbb E\left[\int_H \mathcal X(x)g(x)\mu(dx)\right]+\int_H g^2(x)\mu(dx) \\
        & \geq 
        \mathbb E\left[\int_H\mathcal X^2(x)\mu(dx)\right]-2\left(\mathbb E\left[\int_H \mathcal X(x)^2\mu(dx)\right]\right)^{1/2}\, \norm{g}_{L^2(\mu)}  + \norm{g}_{L^2(\mu)}^2  
    \end{split}
\end{equation*}
and thus $I(g)\to\infty$ as $\norm{g}_{L^2(\mu)}\to \infty$, i.e. $I$ is coercive on $L^2(\mu)$. Moreover, from the above we see that
$$
I(g) - I(\overline{V}) = \norm{g - \overline{V}}_{L^2(\mu)}^2,\quad g \in L^2(\mu).
$$
Therefore, we get the following "sharp" approximation error: given an arbitrary $\varepsilon>0$, by definition of $\inf$, we may find $g_\varepsilon\in\mathfrak{N}(\sigma)$ such that $I(g_\varepsilon)-I(\overline{V})<\varepsilon^2$. This in turn produces
$$
\norm{g_\varepsilon - \overline{V}}_{L^2(\mu)} < \varepsilon.
$$
Thus, if we have found a neural network $g_\varepsilon$ whose energy $I$ is at most $\varepsilon^2$ away from that of the infimum, then this $g_\varepsilon$ is away from $\overline{V}$ at most $\varepsilon$ in the $L^2$-sense. 

Following the idea of Beck et al. \cite{Beck2021}, the procedure now is to substitute any $g$ with a neural network approximation, and minimize over these instead (we can assess this error as the error between a given $g$ and an approximated neural net). As neural nets approximate over compacts, we let $\mu$ be a measure supported on a compact in $H$, let us say $K\subset H$. Then, we fit the network by virtue of 
\begin{equation}
\label{eq:inf-problem-nn}
\inf_{g\in\mathfrak{N}(\sigma)}\mathbb E\left[\int_H\vert \mathcal X(x)-g(x)\vert^2\mu(dx)\right] \equiv \inf_{g\in\mathfrak{N}(\sigma)}I(g).
\end{equation}
If $\mu(K)<\infty$, we can scale it to achieve a probability measure on $K$, so that we can draw random samples $x^{(1)},\ldots,x^{(N)}$ from it. Equally, randomly drawing $M$ samples of $\mathcal X(x^{(i)})$ for each $x^{(i)}$, denoted $\chi^{(j)}(x^{(i)})$, we find
\begin{equation}
\label{eq:stoch-gradient-descent}
\inf_{g\in\mathfrak{N}(\sigma)}\frac1{N+M}\sum_{j=1}^M\sum_{i=1}^N\vert \chi^{(j)}(x^{(i)})-g(x^{(i)})\vert^2
\end{equation}
The underlying intuition here is that we are using the strong law of large numbers applied to the probability space $(\Omega\times K,\P \otimes \mu)$ in order to approximate the minimization problem \eqref{eq:inf-problem-nn} by 
\eqref{eq:stoch-gradient-descent}. Alternatively, one could sample jointly from $\mu$ and $\mathcal X$, reducing the double summation above to just one. In fact, this is what we do in the numerical studies.  

If we have given an ONB $\{e_i\}_{i\in\mathbb N}$  on $H$, then we can define a compact set $K_L$ as
\begin{equation}\label{compact:set:K}
K_L:=\{x\in H; x=\sum_{k=1}^L x_k e_k, x_k\in[a_k,b_k], a_k<b_k],
\end{equation}
and a natural measure $\mu$ is the uniform measure on $K_L$. We trivially extend this to the whole Borel $\sigma$-algebra of $H$. We can accomplish this because $\mathcal{B}(K_L)=\mathcal{B}(H)\cap K_L$.

\begin{remark}
In derivatives markets, one is concerned with the Greeks, that is, the sensitivities of the option price with respect to changes in the underlying. Such sensitivities are important in hedging questions. If we have trained a neural network $g\in\mathfrak{N}(\sigma)$ to the pricing functional $\overline{V}$, we have accessible fast and efficient ways to derive sensitivities with respect to structural changes in the current term structure. The so-called Delta of an option is the derivative of the option price with respect to the underlying asset price. In our context, the "asset price" is the complete term structure curve, and we can be interested in many different "Deltas", like for example a shift in the term structure. If the first basis function $e_1\in H$ measures the level of the term structure, then, if the activation function is (Fr\'echet) differentiable, we can find an analytic expression for $\partial_{1}g(x)$, where $\partial_1$ is the derivative with respect to the first coordinate of $x$. This provides us with an efficiently computable approximation of the "level-Delta" of the option.   
\end{remark}

\section{Numerical pricing by feedforward neural networks in Hilbert space}\label{sec:case:study}
We now test the numerical procedure introduced in Section~\ref{sec:approx:price:fct} for pricing options on the forward curve. As a state space we use the Filipovi\'c space, denoted by $H_w$, introduced in Example~\ref{exp:Filipovic:space}.  

We will first need to derive a basis for this space. For this, let $\{\tilde{e}_i\}_{i\in \mathbb{N}}$ defined by $\tilde{e}_1(\tau)=1$, and for $k>1$
\begin{equation}
\label{eq:def-indep-vectors}
\tilde{e}_i(\xi)=\xi^{(i-2)}\exp(-\xi).
\end{equation}
Then $\{\tilde{e}_i\}_{i\in \mathbb{N}}$ is a set of linearly independent vectors in $H_w$. Observe that the first vector is the {\it level} of the term structure, while the second and third vectors can be associated with {\it slope} and {\it curvature}.

In order to derive from $\{\tilde{e}_i\}_{i\in \mathbb{N}}$ a set of orthonormal vectors we specify the weight function $w(\xi):=\exp( \xi)$ in $H_w$.
We can now apply the Gram-Schmidt algorithm to obtain an orthonormal basis. In the next Lemma we present the 7 first basis vectors resulting from this algorithm (which will be used in our numerical studies): 
\begin{lemma}
From $\{\tilde{e}_i\}_{i\in\mathbb N}$ in \eqref{eq:def-indep-vectors} the first 7 vectors in the orthonormal basis $\{e_i\}_{i\in\mathbb N}$ derived by the Gram-Schmidt algorithm are
\begin{align*}
e_1(\xi)&=1 \\
e_2(\xi )&=\exp(-\xi) -1 \\
e_3(\xi)&=\xi\exp(-\xi) \\
e_4 (\xi)&= \frac{1}{2} \left( \xi^2-2  \xi\right)e^{-\xi} \\
e_5 (\xi)&= \frac{\xi^3-36 \xi^2-6 \xi}{42 \sqrt{5}} e^{-\xi}\\
e_6 (\xi)&= \frac{ \xi^4-1440 \xi^3-192 \xi^2-24 \xi}{24 \sqrt{806115}}e^{-\xi}\\
e_7 (\xi)&= \frac{ \xi^5-100800 \xi^4-10800 \xi^3-1200 \xi^2-120 \xi}{1560 \sqrt{49407661}}e^{-\xi}.
\end{align*}
\end{lemma}
\begin{proof}
First, recall that 
$$\int_0^{\infty} \xi^k \exp (-\xi)\,d\xi=\Gamma (k+1)=k !.$$
For $i=1$, let $\hat{e}_1 :=\tilde{e}_1=1$ and using that $\vert\hat{e}_1\vert_w=1$, we define $e_1=\hat{e}_1/\vert\hat{e}_1\vert_w=1$. Next we calculate 
 \begin{eqnarray}
    \hat{e}_2 &:=& \tilde{e}_2 - \Pi_{\hat{e}_1}(\tilde{e}_2)=\tilde{e}_2 - \langle \tilde{e}_2, \hat{e}_1 \rangle_w \hat{e}_1=\tilde{e}_2 - \hat{e}_1 \nonumber
    \end{eqnarray}
and using that $\vert\hat{e}_2\vert_w^2=\int_0^{\infty} \exp (-\xi) d \xi=\Gamma (0+1)=1$ we define $$e_2(\xi ):=\hat{e}_2(\xi )/\vert\hat{e}_2\vert_w=\exp(-\xi) -1.$$
For $i=3$,
       \begin{eqnarray}
    \hat{e}_3 &:=& \tilde{e}_3 - \Pi_{\hat{e}_1}(\tilde{e}_3)-\Pi_{\hat{e}_2}(\tilde{e}_3)=\tilde{e}_3 - \langle \tilde{e}_3, \hat{e}_1 \rangle_w \hat{e}_1 - \langle \tilde{e}_3, \hat{e}_2 \rangle_w \hat{e}_2\nonumber
    \end{eqnarray}
    which gives $\hat{e}_3 (\xi)=\xi \exp(-\xi)$. Furthermore, $\vert \hat{e}_3 \vert_w^2= \Gamma (2+1) -2 \Gamma (1+1) +\Gamma (0+1)=1$, and thus $e_3:=\hat{e}_3$. Continuing with this procedure, tedious technical calculations yield the remaining 4 basis vectors.
\end{proof}
See Figure~\ref{basis} for a plot of the functions $\tilde{e}_1,...,\tilde{e}_{7}$ and the orthonormal functions $e_1,...,e_{7}$. 

\begin{figure}[t]
\hfill\subfigure[]{\includegraphics[width=0.45\textwidth]{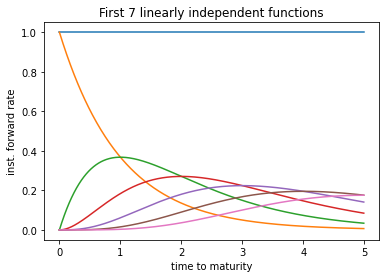}\label{lin:indep:fct}} 
    \hfill \subfigure[]{\includegraphics[width=0.45\textwidth]{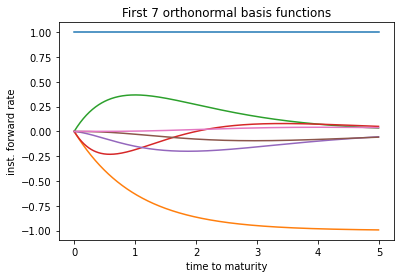}\label{basis:fct}} 
\caption{(a) Functions $\tilde{e}_1,\dots,\tilde{e}_7$. (b) Functions $e_1,\dots,e_7$}
\label{basis}
\end{figure}

We consider a specification of the exponential model (\ref{eq:explicit-price-non-state-dep}) and (\ref{eq:markovian-struct-exponential}) for the forward curve, i.e., we assume that 
the instantaneous forward curve is described by 
$X(t)=\exp(Y(t))$ with $Y$ as in \eqref{eq:markovian-struct-exponential}. The price of the instantaneous forward delivering at the fixed time $T$ can then be obtained by $F(t,T)=\delta_{T-t}X(t)$. We consider a call option on a flow forward with strike $\mathcal{K}$, that is, we introduce the payoff function $\mathfrak{P}(y):=\max \{ z-\mathcal{K}, 0\}$. Recalling the flow forward price $\hat{F}(t,T_1,T_2)$ in \eqref{def:flow-forward}, the option pays at time $\tau\leq T_1$ the amount $\mathfrak{P}(\widehat{F}(\tau,T_1,T_2)$. Following the derivations in Sections \ref{sec:option} and \ref{sec:approx:price:fct}, we recall the price as $\overline{V}(x):=\mathbb{E} [\mathcal{X}(x)]$ with $\mathcal X(x)$ defined in \eqref{eq:X-def-geom}. In our numerical case studies, we shall fix our attention to options on monthly flow forwards, with exercise time at the beginning of the delivery period. More specifically, we assume $\tau=T_1=1/12$ and $T_2=2/12$, with time being measured in months. We price the option at current time $t=0$ with a strike of $\mathcal{K}=1$. We set the risk free interest rate $r=0$.

Following Section~\ref{sec:approx:price:fct} we want to obtain a numerical approximation of $\overline{V} : H_w \rightarrow \mathbb{R}$. As described in Section~\ref{sec:approx:price:fct} we can only expect to approximate the true pricing function $\overline{V}$ well on some compact subset of $H_w$. As proposed in (\ref{compact:set:K}) we choose $K_{7}\subset H$ by 
$$K_{7}:=\{x\in H_w; x=\sum_{k=1}^{7} x_i e_i, x_i\in[-1/2,1/2] \}.$$ The functions $e_8,e_9,\dots $ have very low absolute value in the short end ($\xi \leq 2/12$,  corresponding to the first two month), which is the reason why we only use the subspace spanned by $e_1,\dots , e_7$. Let $\mu$ be the uniform measure on $K_7$.

We now describe the practical implementation of the neural network architecture explained in Section~\ref{nn:frechet}. By Proposition~\ref{finite:dim:projection:NN} we know that for any given $\delta>0$ and a discriminatory $\sigma \in C(H_w,H_w)$ there exists $N_\ast=N_\ast(p, \sigma,\delta)\in\N, M \in \mathbb{N}$ and $A_j \in \mathcal{L}(H_w), b_j \in H_w$ for $j=1,\dots , M$ such that for any $N\geq N_\ast$ 
\begin{equation*}
    \sup_{x\in K}\abs{\overline{V}(x)-\sum_{j=1}^M\langle \ell_j\circ\Pi_N,\sigma(
    \Pi_{N}A_j\Pi_{N}x+\Pi_{N}b_j)\rangle}< \delta,
\end{equation*}
where $\Pi_N : \mathcal{X}\rightarrow \Span \{ e_1, \dots , e_N \}$ is the orthogonal projection. Let us now fix such $N\geq \max\{N_\ast, 7\} $. We define $\bar{A}_j$ to be the restriction of $\Pi_{N}A_j$ to $H_w^N:=\Span \{ e_1, \dots , e_N \}$ and $\bar{b}_j:=\Pi_{N}b_j$ for $j=1,\dots , M$. Then, clearly $\bar{A}_j \in \mathcal{L} (H_w^N)$ and $\bar{A}_j$ can be associated with a matrix $\mathbb{R}^{N\times N}$. Further $\bar{b}_j \in H_w^N$ and it can be associated with an element in $\mathbb{R}^{N}$. We further define $\bar{\ell}_j\in \mathcal{L} (H_w^N,\mathbb{R})$ for $j=1,\dots , M$ to be the restriction of $\ell_j$ to $H_w^N$. If now our activation function $\sigma $ has an image $\text{Im}(\sigma)$ such that $\text{Im}(\sigma)\subset H_w^N$, then the restriction of $\sigma$ to $H_w^N$ defines an element $\bar{\sigma}\in C(H_w^N,H_w^N)$. We obtain that for $x\in H_w^N$ it holds that 
\begin{equation}
\langle \ell_j\circ\Pi_N,\sigma( \Pi_{N}A_j\Pi_{N}x+\Pi_{N}b_j)\rangle=\langle \bar{\ell}_j ,\bar{\sigma} ( \bar{A}_j x+\bar{b}_j)\rangle.
\end{equation}
In particular all quantities appearing on the right hand side are finite dimensional. After projecting onto the subspace $H_w^N$, the one layer neural network becomes: 
\begin{equation}
\sum_{j=1}^M\langle \bar{\ell}_j ,\bar{\sigma} ( \bar{A}_j x+\bar{b}_j)\rangle.
\end{equation}

In Figure~\ref{NN} we present a picture of the neural network architecture. We implement this neural network in Python using the TensorFlow and Keras libraries. The code is available on GitHub\footnote{\href{https://github.com/ncdetering/FlowForwardsNumerics}{https://github.com/ncdetering/FlowForwardsNumerics}}. In contrast to classical neural networks where every node has an output dimension of $1$, in our network every node in the main affine layer computes the map $\bar{A}_j x+\bar{b}_j$ with output dimension $N$. In our implementation we therefore represent each of these nodes as a dense layer. Each of these dense layers then has input and output dimension equal to $N$ and is connected to every node in the input layer. This means that the hidden layer depicted in Figure~\ref{NN} actually consists of $M$ vertically arranged dense layers, each with input and output dimension equal to $N$. Similarly each node in the activation layer receives input of dimension $N$ and produces output of dimension $N$. We use $M$ vertically arranged Keras Lambda layers to implement this tailor made activation layer. The nodes in the linear forms receive $N$ dimensional input and produce $1$ dimensional output. Internally this is implemented as $M$ vertically arranged dense layers with input dimension $N$ and output dimension equal to $1$. The single node in the last layer is just implementing the summation. 

The activation function $\sigma$ that we use follows Example 4.4 in Benth, Detering and Galimberti \cite{benth2021neural}. Specifically, we choose it of the form $\sigma (x)=\beta(\psi (x)) z$ for some map $\psi \in \mathcal{L} ( H_w, \mathbb{R})$, $z\in H_w$ and $\beta\in Lip(\R;\R)$ such that $\lim_{\xi\to\infty} \beta(\xi)=1,\;\lim_{\xi\to-\infty} \beta(\xi)=0$ and $\beta(0)=0$. To ensure that $\text{Im}(\sigma)\subset H_w^N$ we choose $z\in H_w^N$ but we note that we could start with a general $z\in H_w$ (before determining $N$) and then after determining $N$ for the finite network implementation, we replace it with $\Pi_N (z)$. We further specify $\beta(y):= \max \{0,1-\exp (-y) \}$. Because $\bar{\sigma}=\beta (\psi (\Pi_N (x)))\Pi_N (z)$ we can interpret $\psi$ as a map in $\mathcal{L} ( \mathbb{R}^N, \mathbb{R})$ that acts on the first $N$ coefficients of $x$.

 \begin{figure}[t]\label{NN}
	\centering
	\begin{tikzpicture}[shorten >=1pt]
		\tikzstyle{unit}=[draw,shape=circle,minimum size=1.15cm]
		\tikzstyle{hidden}=[draw,shape=circle,minimum size=1.15cm]
 
		\node[unit](x0) at (0,3.5){$x_1$};
		\node[unit](x1) at (0,2){$x_2$};
		\node at (0,1){\vdots};
		\node[unit](xd) at (0,0){$x_N$};
 
		\node[hidden](h10) at (3,4.5){$\bar{A}_1 \mathbf{x} + \bar{b}_1$};
		\node[hidden](h11) at (3,2.5){$\bar{A}_2 \mathbf{x} + \bar{b}_2$};
		\node at (3,1.5){\vdots};
		\node[hidden](h1m) at (3,-1){$\bar{A}_M \mathbf{x} + \bar{b}_M$};
 
		\node(h22) at (5,0){};
		\node(h21) at (5,2){};
		\node(h20) at (5,4){};
		
 
		
		\node[hidden](hL0) at (6,4.5){$\bar{\sigma}$};
		\node[hidden](hL1) at (6,2.5){$\bar{\sigma}$};
		\node at (6,1.5){\vdots};
		\node[hidden](hLm) at (6,-1){$\bar{\sigma}$};
		
		\node[hidden](l0) at (9,4.5){$\bar{\ell}_1$};
		\node[hidden](l1) at (9,2.5){$\bar{\ell}_2$};
		\node at (9,1.5){\vdots};
		\node[hidden](lm) at (9,-1){$\bar{\ell}_M$};
 
		\node[unit](sum) at (12,2){$\sum$};

		\draw[->,red] (x0) -- (h10);
		\draw[->,red] (x0) -- (h11);
		\draw[->,red] (x0) -- (h1m);
 
 		\draw[->,red] (x1) -- (h10);
 		\draw[->,red] (x1) -- (h11);
		\draw[->,red] (x1) -- (h1m);
 
  		\draw[->,red] (xd) -- (h10);
		\draw[->,red] (xd) -- (h11);
		\draw[->,red] (xd) -- (h1m);
 
		\draw[->,blue] (hL0) -- (l0);
 
		\draw[->,blue] (hL1) -- (l1);
 
		\draw[->,blue] (hLm) -- (lm);
 
		\draw[->,blue] (h10) -- (hL0);
		
		\draw[->,blue] (h11) -- (hL1);
		
		\draw[->,blue] (h1m) -- (hLm);
		
		 \draw[->] (l0) -- (sum);
		  \draw[->] (l1) -- (sum);
		   \draw[->] (lm) -- (sum);
		
		
		
		\draw [decorate,decoration={brace,amplitude=10pt},xshift=-4pt,yshift=0pt] (-0.5,4) -- (0.75,4) node [black,midway,yshift=+0.6cm]{input layer};
		\draw [decorate,decoration={brace,amplitude=10pt},xshift=-4pt,yshift=0pt] (2.5,5.3) -- (3.75,5.3) node [black,midway,yshift=+0.6cm]{affine layer};
		\draw [decorate,decoration={brace,amplitude=10pt},xshift=-4pt,yshift=0pt] (5.5,5) -- (6.75,5) node [black,midway,yshift=+0.6cm]{activation layer};
		\draw [decorate,decoration={brace,amplitude=10pt},xshift=-4pt,yshift=0pt] (8.5,5) -- (9.75,5) node [black,midway,yshift=+0.6cm]{linear form};
		\draw [decorate,decoration={brace,amplitude=10pt},xshift=-4pt,yshift=0pt] (11.5,2.8) -- (12.75,2.8) node [black,midway,yshift=+0.6cm]{output layer};
	\end{tikzpicture}
	\caption{Illustration of the structure of the neural network. A red arrow between two nodes means that one dimensional information at the tail of the arrow (here the values $x_k$) is carried to a node at the head of the arrow in which this one dimensional information is internally processed by a function of output dimension $N$. A blue arrow between two nodes means that it carries $N$ dimensional information from the tail to a node that has also $N$ dimensional output dimension. A black arrow corresponds to the standard setting where one dimensional information is carried to a node that has a one dimensional output.}
	\label{fig:multilayer-perceptron}
\end{figure}
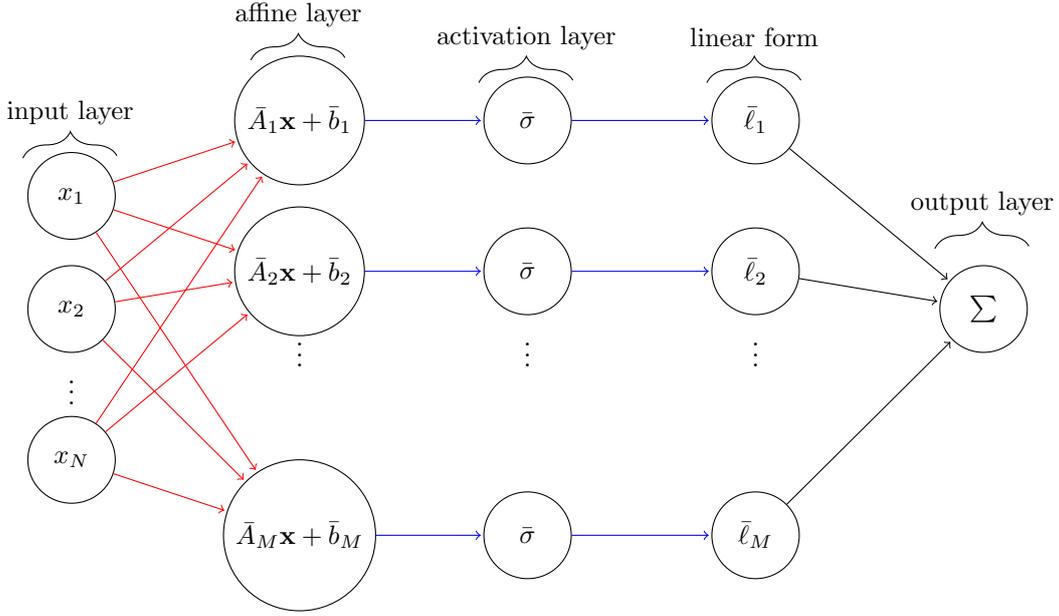
We consider two specifications for the coefficients and for the Wiener process in \eqref{eq:explicit-price-non-state-dep} defining the exponential model in (\ref{eq:markovian-struct-exponential}):

\subsection*{One dimensional noise:} In this first example we consider a setup with a simple one dimensional noise. To specify the coefficients in the mild solution \eqref{eq:explicit-price-non-state-dep}, we choose $W(t)=B(t) e_1$
with $B$ a standard Brownian motion. Then it follows for the covariance operator $\mathcal{Q}$ of $W$ that $\mathcal{Q}=\Pi_{H^1_w}$ where $H^1_w=\text{span} \{ e_1 \}$. Moreover we choose $\eta (s)= \textbf{id}$. Then $\alpha  (s)=-(1/2) e_1$ ensures that we obtain an arbitrage free model for the forward curve dynamics.

Then, for $x=\sum_{k=1}^{7} x_k e_k$ we get that 
$$\mathcal{S}_t x=x_1 e_1(\cdot) + \sum_{k=2}^{7} x_k e_k (\cdot +t)$$
and thus $Y(t)(\xi)=x_1 -(1/2) t + B(t)-B(0)+ \sum_{k=2}^{7} x_k e_k(\xi+t)$.

\noindent 
\subsection*{Multi dimensional noise:} Let the Wiener process now be given by 
$$W(t)=B_1(t) e_1 + \dots + B_7(t) e_7$$ with $B_i, 1 \leq i\leq 7$ independent standard Brownian motions and let $\eta=\eta^*=\textbf{id}$.  We need to calculate the drift $\alpha$ that leads to an arbitrage free model. According to (\ref{arb:free:drift}) the drift is $$\alpha(t,\xi)=-\frac12\vert \mathcal Q^{1/2}\eta^*(t)(\delta_{\xi}^*1)\vert_w^2 .$$ 
By definition of $W$, we then get that $\mathcal{Q}=\mathcal{Q}^{1/2}=\textbf{id}\circ \Pi_{K_7}=\Pi_{K_7}$ with $K_7=\text{span} \{e_1, \dots , e_7\}$ and thus by $\delta_{\xi}^*1=h_{\xi}$ we get that 
$$\alpha(t,\xi)=-\frac12\vert  \Pi_{K_7} (h_{\xi})\vert_w^2 .$$
with
$$ h_{\xi}(u)=1+\int_0^{\xi\wedge u}w^{-1}(v)dv .$$
Because $h_{\xi}$ is the representation for $\delta_{\xi}$, it follows that $\Pi_{K_7} (h_{\xi})=e_1 (\xi)\cdot e_1 + \dots + e_7 (\xi)\cdot e_7$ and therefore 
by orthonomality of $e_1,\dots, e_7$
$$\alpha(t,\xi)=-\frac12\vert  \Pi_{K_7} (h_{\xi})\vert_w^2 =-\frac12 ( e_1 (\xi)^2 + \dots + e_7 (\xi)^2).$$
It follows that 
\begin{eqnarray}
Y(t)&=&\mathcal S_tY_0+\int_0^t\mathcal S_{t-s}\alpha ds+\int_0^t\mathcal S_{t-s}\eta dW(s) \nonumber \\
&=& Y_0 (\cdot + t)+\int_0^t  \alpha (\cdot + (t-s)) ds+\int_0^t\mathcal S_{t-s} dW(s) \nonumber \\
&=& Y_0 (\cdot + t)- \frac12 \int_0^t e_1 (\cdot + (t-s))^2 + \dots + e_7 (\cdot + (t-s))^2 ds\nonumber\\
&  & +\sum_{i=1}^7 \int_0^t e_i (\cdot + (t-s)) dB_i (s)\nonumber
\end{eqnarray}
We are pricing an option with maturity in one month and one month delivery. So we need to simulate from $Y(1/12)$. For this we fix a discretization size $L$ and put $0=s_0 < s_1 < \dots < s_L=t$ and approximate  
$$- \frac12 \int_0^t e_1 (\cdot + (t-s))^2 + \dots + e_7 (\cdot + (t-s))^2 ds \approx - \frac12 \sum_{i=1}^7 \sum_{j=1}^L e_i (\cdot + (t-s_j))^2 (s_j - s_{j-1})$$

$$\sum_{i=1}^7 \int_0^t e_i (\cdot + (t-s)) dB_i (s)\approx \sum_{i=1}^7 \sum_{j=1}^L e_i (\cdot + (t-s_j))(B_i (s_j) - B_i (s_{j-1}) ).$$



We then generate our training set of size $n=10,000,000$. For this let $$x^{(k)}=(x^{(k)}_1,\dots , x^{(k)}_{7} ) \in [-1/2,1/2]^7 $$ for $k=1,\dots , n$ be the i.i.d. realizations of coefficients for our initial curves so that each $x^{(k)}$ can be seen as sampled from $\mu$.

We now compute the payoff $\mathcal{X}(x^{(k)})$ of the option as defined in \eqref{eq:X-def-geom} under both the one- and multi-dimensional noise models described above based on one realization of the respective Wiener processes $W$. So $\mathcal{X}(x^{(k)})$ is computed as the time $\tau=1/12$ value under the model (\ref{eq:explicit-price-non-state-dep}) and (\ref{eq:forward-HJM-geometric}) with the initial curve (starting value) $x^{(k)}$ and based on a realization of the process $W$ that is independent of the starting values $x^{(k)}$. We stress that by $x^{(k)}$ we actually mean the vector $\sum_{i=1}^7x_i^{(k)}e_i\in H_w$. Our training set then consists of the input-output pairs $(x^{(k)}, \mathcal{X}(x^{(k)}) \in [-1/2,1/2]^7\times \mathbb{R}$ for $k=1,\dots , n$. 

We use this training set to fit several neural networks with the number of hidden nodes ranging from $1$ to $30$. For the training with stochastic gradient descent we use a batch size of $10,000$ and $50$ epochs. We then generate a test set of size $10,000$. In contrast to the training set, the test set is composed of pairs consisting of starting values $x_{\text{test}}^{(k)}$, and option prices $\mathbb{E} [\mathcal{X}(x_{\text{test}}^{(k)})]$ for $k=1,\dots , 10,000$. The calculation of the benchmark prices $\mathbb{E} [\mathcal{X}(x^{(k)})]$ is based on a Monte Carlo simulation with $100,000$ simulations each. We then calculate for each of the fitted models the mean squared error with respect to the test set. 

The results are shown in Figure~\ref{MSE:network:size} for both, the model with one dimensional noise, and the model with $7$ dimensional noise. Except for networks with a very small number of nodes ($\leq 5$), the mean squared error is of the order $10^{-5}$ both for the specification with one dimensional noise and for the one with $7$ dimensional noise. For the one dimensional noise setting the average error over the networks with the number of nodes ranging from $10$ to $30$ is $2.047 \cdot 10^{-5}$ and for the multi dimensional noise it is $1.859\cdot 10^{-5}$. This error is fluctuating due to randomly assigned initial network weights but is always within the range $1.5\cdot 10^{-5}$ to $2.5\cdot 10^{-5}$.

\begin{figure}[t]
\hfill\subfigure[]{\includegraphics[width=0.45\textwidth]{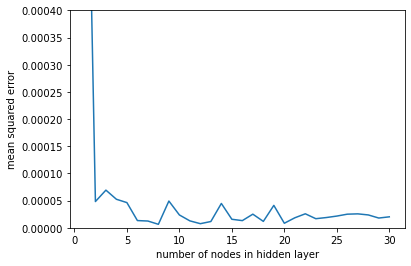}\label{MSE:network:size:1dim}} 
    \hfill \subfigure[]{\includegraphics[width=0.45\textwidth]{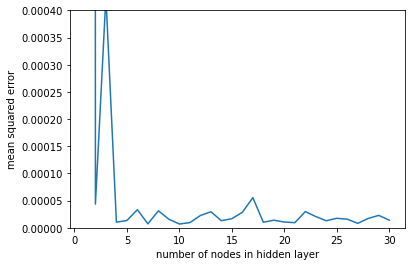}\label{MSE:network:size:7dim}} 
\caption{Mean squared error for network architecture presented in Section~\ref{nn:frechet} with number of neurons ranging from $1$ to $30$ (a) One dimensional noise specification. (b) Multi dimensional noise specification.}
\label{MSE:network:size}
\end{figure}
In Figure~\ref{x1approx} we plot the option price for initial log forward curves with varying basis coefficient $x_1$ and all other coefficients are set to $0$. The green line describes the option price as determined by the neural network while the red stars are option prices as determined by Monte Carlo simulations. Similarly in Figure~\ref{x2approx} we vary the coefficient $x_2$ instead and keep the other coefficients equal to $0$. While for the first coefficient the fit seems almost perfect, for the second coefficient the fit looks worse. This is due to the lower option price. Because in this case $x_1=0$, the options are not as deep in the money and the option price depends more on the tail of the distribution of $W$. To achieve a better fit for these options one would have to increase the training size to include more samples of $W$. One can also observe that the red stars still show some noisy behaviour due to the limited number of only $100,000$ simulations. 
\begin{figure}[t]
    \hfill \subfigure[]{\includegraphics[width=0.45\textwidth]{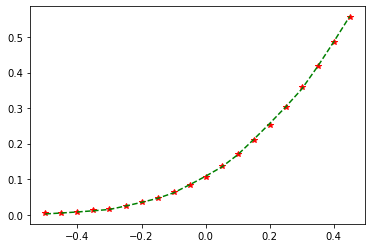}\label{e1approx:Uni}} \hfill\subfigure[]{\includegraphics[width=0.45\textwidth]{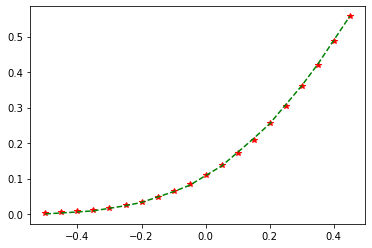}\label{e2approx:Uni}} 
\caption{(1) (a) Option price with varying coefficient $x_1$ for the model with one dimensional noise. (b) Option price with varying coefficient $x_1$ for the model with multi dimensional noise. The green line marks the option price determined by the neural network. The red stars are option prices determined by Monte Carlo simulations.}\label{x1approx}
\end{figure}
\begin{figure}[t]
    \hfill \subfigure[]{\includegraphics[width=0.45\textwidth]{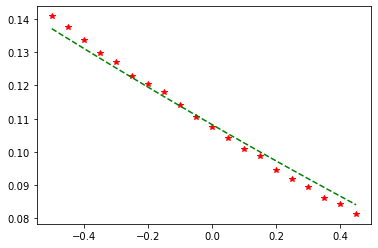}\label{e1approx:Multi}} \hfill\subfigure[]{\includegraphics[width=0.45\textwidth]{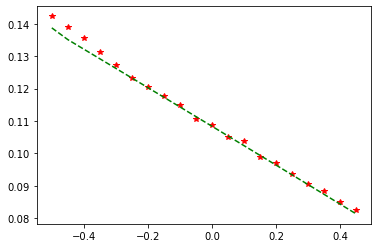}\label{e2approx:Multi}} 
\caption{(1) (a) Option price with varying coefficient $x_2$ for the model with one dimensional noise. (b) Option price with varying coefficient $x_2$ for the model with multi dimensional noise. The green line marks the option price determined by the neural network. The red stars are option prices determined by Monte Carlo simulations.}\label{x2approx}
\end{figure}

We next compare the accuracy of our Hilbert space neural network architecture to the one obtained with a classical neural network with one dimensional activation and where instead of using the basis coefficients as inputs to the network we discretize the initial function $x$ on an equally spaced grid. For this, recall that each sample $x^{(k)}$ represents the function $x^{(k)}(\xi) = \sum_{i=1}^7 x^{(k)}_i e_i (\xi)$. We evaluate the function on a grid of equally spaced points $0=\xi_0 < \dots < \xi_{D}$ and analyze the accuracy with a standard one layer neural network with input dimension $D$ trained with the input-output pairs $(x^{(k)}(\xi_0), \dots , x^{(k)}(\xi_D)),\mathcal{X}(x^{(k)}))$. We consider $D=10$ and $D=20$. We fit neural networks with the number of hidden nodes ranging from $5$ to $150$. For $D=10$ the number of parameters to be fitted for these networks ranges from $60$ ($5$ nodes) to $1,800$ ($150$ nodes) which is comparable to the previous setting where the number of parameters ranged from $63$ ($1$ node) to $1,890$ ($30$ nodes). For $D=10$ the parameters range from $110$ ($5$ nodes) to $3,300$ ($150$ nodes).
We display the resulting mean squared errors in Figure~\ref{mse:classic}. As one can see in the figure, the mean squared error is significantly larger than with the Hilbert space neural network. Averaging over the networks with at least $20$ nodes, the average mean square error is $1.251 \cdot 10^{-3}$ (gridsize $10$) and $1.550\cdot 10^{-3}$ (gridsize $20$) for the one dimensional noise, and $1.075 \cdot 10^{-3}$ (gridsize $10$) and $1.827 \cdot 10^{-3}$ (gridsize $20$) for the multi dimensional noise. 
\begin{figure}[t]
    \hfill \subfigure[]{\includegraphics[width=0.45\textwidth]{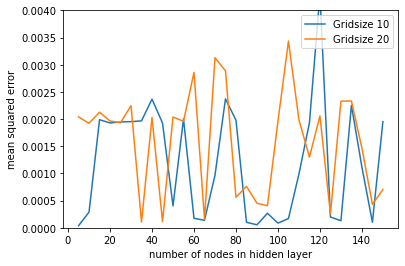}\label{classic:Uni}} \hfill\subfigure[]{\includegraphics[width=0.45\textwidth]{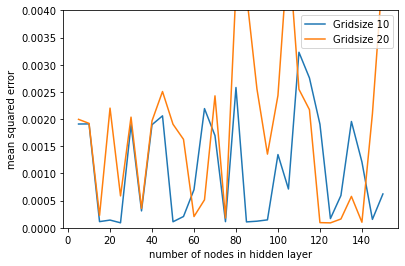}\label{mse:classic:Multi}} 
\caption{Mean squared error for classical neural network architecture with number of neurons ranging from $5$ to $150$ (a) One dimensional noise specification. (b) Multi dimensional noise specification.}\label{mse:classic}
\end{figure}

\subsection*{Run-time} We perform our numerical analysis on a virtual machine provided by Letters \& Science IT at the University of California, Santa Barbara. The virtual machine provides 32 CPU cores with speed $2.1$GHz. The generation of training set and test set is done parallel across all cores by using the python multiprocessing package. The training of the neural networks facilitates the standard parallel computing capability of the TensorFlow library. The generation of the training set takes less than $1$ minute for the specification with one dimensional noise and $27$ minutes for the specification of multi dimensional noise. Generation of the test set requires the most computational resources due to the Monte Carlo simulation performed for each initial curve in the training set resulting in a total of $1$ billion simulations. For the model with one dimensional noise, this takes $1$ hour and $33$ minutes, and for the model with multidimensional noise it takes $43$ hours and $35$ minutes. The reason for the significantly longer run-time of the multi dimensional setting is mainly the discretisation that is required for simulating from $Y(t)$ in this case. In the one dimensional noise setting this is not necessary because the Wiener process is of the particular form $W(t)=B(t)e_1$ where $e_1$ is the function constantly equal to $1$ on which the semi group $S_t$ acts as the identity. The training time of the neural networks is independent of the specification of the noise and takes from $5$ minutes to $24$ minutes where the time increases with the number of parameters of the network. 

In conclusion, run time is not a limiting factor for turning this pricing routine into production for two reasons:
\begin{enumerate}
    \item The computing resources used for this study are limited and comparable to a modern laptop computer. They can be easily scaled up by a financial institution. Moreover, the generation of the training and test set, could be programmed in C++, and be compiled into executable code that does not require interpretation at run time. This leads to a significant speed improvement over the Jupyter notebooks used in this study which are interpreted at run time.
    \item We derive the pricing function automatically for all initial curves in $K_7$. For that reason, it is not necessary to retrain the network when the market changes. In fact, generating a new training and test set, and retraining the network can be done either over night or over the weekend. This is a common practice in banks for pricing routines that are more complex. 
\end{enumerate}


\bibliographystyle{abbrv}
\bibliography{literature}

\end{document}